\documentclass[pra,a4paper,notitlepage,nofootinbib,superscriptaddress]{revtex4-1}
\pdfoutput=1

\usepackage[verbose=true]{microtype}
\usepackage[utf8]{inputenc}
\usepackage[colorlinks,citecolor=black,linkcolor=black,linktocpage,breaklinks,hypertexnames=false,pdfpagelabels,draft=false]{hyperref}
\usepackage[intlimits]{amsmath}
\usepackage{amssymb,amsfonts,amsthm}
\usepackage{enumerate}
\usepackage[capitalise]{cleveref}

\theoremstyle{plain}
\newtheorem{lemma}{Lemma}
\newtheorem{theorem}[lemma]{Theorem}
\newtheorem{proposition}[lemma]{Proposition}
\newtheorem{corollary}[lemma]{Corollary}
\newtheorem{definition}[lemma]{Definition}

\renewcommand{\epsilon}{\varepsilon}
\newcommand{\E}{\mathcal{E}}
\newcommand{\N}{\mathbb{N}}

\newcommand{\1}{{\openone}}

\newcommand{\be}{\begin{eqnarray}}
\newcommand{\ee}{\end{eqnarray}}
\newcommand{\ba}{\begin{array}}
\newcommand{\ea}{\end{array}}

\newcommand{\ra}{\rangle}
\newcommand{\la}{\langle}
\newcommand{\mc}{\mathcal}

\DeclareMathOperator{\tr}{Tr}

\newcommand{\diag}{\mathrm{diag}}

\newcommand{\ben}{\begin{enumerate}}
\newcommand{\een}{\end{enumerate}}
\newcommand{\bi}{\begin{itemize}}
\newcommand{\ei}{\end{itemize}}

\begin{document}

\title{Irreducible forms of Matrix Product States: Theory and Applications}

\author{Gemma De las Cuevas}
\address{Institut f\"ur Theoretische Physik, Universit\"at Innsbruck, 
Technikerstr.\ 21a, 6020 Innsbruck, Austria} 
\author{J. Ignacio Cirac}
\address{Max Planck Institute for Quantum Optics, 
Hans-Kopfermann-Str.\ 1, 85748 Garching, Germany}
\author{Norbert Schuch}
\address{Max Planck Institute for Quantum Optics, 
Hans-Kopfermann-Str.\ 1, 85748 Garching, Germany}
\author{David Perez-Garcia}
\address{Departamento de An\'alisis Matem\'atico and IMI, 
Universidad Complutense de Madrid, 28040 Madrid, Spain}
\address{ICMAT, C/ Nicol\'as Cabrera, Campus de Cantoblanco, 28049 Madrid,
Spain}

\begin{abstract}
The canonical form of Matrix Product States (MPS) and the associated
fundamental theorem, which relates different MPS representations of a
state, are the theoretical framework underlying many of the
analytical results derived through MPS, such as the classification of
symmetry-protected phases in one dimension.  Yet, the canonical form
is only defined for MPS without non-trivial periods, and thus cannot fully
capture paradigmatic states such as the antiferromagnet.  Here, we
introduce a new standard form for MPS, the \emph{irreducible form}, which
is defined for arbitrary MPS, including periodic states, and show that any
tensor can be transformed into a tensor in irreducible form describing the
same MPS. We then prove a fundamental theorem for MPS in irreducible form:
If two tensors in irreducible form give rise to the same MPS, then they
must be related by a similarity transform, together with a matrix of
phases.
We provide two applications of this result: an equivalence between
the refinement properties of a state and the divisibility properties of
its transfer matrix,  and a more general characterisation of tensors that
give rise to matrix product states with symmetries.  \end{abstract} 

\maketitle

\section{Introduction}
\label{sec:introduction}

The description of quantum many-body systems is not scalable due to the
exponential growth of the Hilbert space dimension with the number of subsystems.
The desire to develop efficient techniques to analyse strongly correlated
systems has motivated the program of tensor networks, which are a
theoretical and numerical tool \cite{Ve08,Or14b} to describe these systems
in various settings (various physical dimensions, with various symmetries,
appropriate for describing ground states of gapped or gapless
Hamiltonians, etc).  The simplest and most thoroughly studied type of tensor
networks are Matrix Product States (MPS)~\cite{Fa92,Sc11d,Pe07}, which are
suitable for describing ground states of one-dimensional gapped
Hamiltonians.  Among many other things, they have allowed to characterise the
symmetries of states in terms of the corresponding tensors
\cite{Pe08}, and have given rise to a classification of gapped phases in
one-dimensional systems~\cite{Ch11,Sc11,Po10}.

One of the most interesting features of MPS is that they allow to describe
families of  translationally invariant systems, even in the thermodynamic
limit, in a simple and concise way. Any rank-three tensor $A$, with
coefficients $A^i_{\alpha,\beta}$ with  $i=1,\ldots,d,$ and
$\alpha,\beta=1,\ldots,D$,  generates a family of translationally
invariant MPS,\footnote{Strictly speaking, this is a family of matrix
product vectors \cite{Ci15}, but we shall refer to them as matrix product
states in this paper.} namely
\be
 \mc{V}(A)=\left\{ 
 |V_N(A)\ra = \sum_{i_1,\ldots, i_N=1}^d\tr(A^{i_1}\cdots A^{i_N})|i_1,\ldots, i_N\ra\in (\mathbb{C}^d)^{\otimes N}\right\}_{N\in \mathbb{N}},
 \label{eq:VA}
 \ee 
  where each $ |V_N(A)\ra$ corresponds to a state of $N$ spins of physical
dimension $d$,  and the $D\times D$ matrices $A^i$ have coefficients
$A^i_{\alpha,\beta}$. Thus, the properties of this whole family of states
are completely determined by the tensor $A$, and therefore by a number of
coefficients that is independent of $N$.  This fact has enabled to study
many-body systems by just analysing the properties of $A$. 

Many of the results obtained for MPS rely on what is called the canonical
form, together with the associated ``fundamental theorem'' which relates different
MPS representations of a family of states~\cite{Ci15}: 
Different tensors $A$ and $B$ can generate the same family of states,
$\mc{V}(A)=\mc{V}(B)$,  which introduces an ambiguity when analysing
many-body states in terms of the tensors that generate them, being
crucial e.g.\ in the analysis of phases under symmetries.  The
canonical form constitutes a specific normal form into which an MPS
tensor can be brought without changing the family of states it generates,
and which has a number of favorable properties.  The fundamental
theorem of MPS then asserts that for any two tensors $A$ and $B$ in
canonical form for which $\mc{V}(A)=\mc{V}(B)$, there exists a simple
local gauge transformation relating them: There exists some invertible
matrix $Y$ such that $A^i=Y B^i Y^{-1}$, and thus, the matrices $A^i$ and
$B^i$ are the same up to a change of basis given by $Y$.

However, the canonical form is not defined for all possible tensors $A$.
Specifically, this excludes translationally invariant MPS which are
superpositions of states with non-trivial periodicity $m>1$, such as the
antiferromagnetic state $|0,1,0,1\ldots \ra + |1,0,1,0\ldots \ra$ with
period $m=2$. The way to deal with such states has hitherto been to block $m$
spins, yielding an MPS which is translationally invariant in a trivial
way, at which point the canonical form and the fundamental theorem can be
applied.  However, the local entanglement structure relating to the
non-trivial periodicity is lost in this procedure, and the physical
properties of the system can change radically (e.g., the antiferromagnet
becomes a ferromagnet when blocking $2$ sites). It is therefore desirable
to have a standard form, together with a fundamental theorem, which are
directly applicable to all translational invariant MPS, including periodic
ones.

In this work, we introduce a new standard form, the \emph{irreducible
form}, which is applicable to all tensors generating translationally
invariant MPS, including $m$-periodic states. We show how to transform an
arbitrary tensor $A$ generating a family of translationally invariant MPS into 
a tensor $B$ in irreducible form which generates the same family of states,
$\mc{V}(A)=\mc{V}(B)$, without the need for blocking.  We then derive a fundamental theorem for MPS in
irreducible form: Given any $A$ and $B$ in irreducible form which generate
the same family of states, $\mc{V}(A)=\mc{V}(B)$, there exist a unitary
matrix $Z$ and an invertible matrix $Y$ such that $ZA^i = YB^i Y^{-1}$ for
all $i$, with $[Z,A^i]=0$ and $\mc{V}(A)=\mc{V}(ZA)$, this is, $A$ and $B$
are related by a local gauge transformation.
 
We moreover provide some applications of this result. The first one is in the
context of renormalization of MPS, and relates the possibility of refining
MPS with the divisibility properties of a trace preserving completely
positive map (i.e.\ a quantum channel)  that is associated to the
corresponding tensor. The second one is in the context of symmetries of
MPS and extends the results of previous works to MPS which are periodic.

The rest of the paper is organized as follows. 
In Section \ref{sec:irr} we introduce the irreducible form and show how to transform any tensor $A$ to its irreducible form. We also study how the irreducible form behaves under blocking.
In Section \ref{sec:thm} we present the fundamental theorem for MPS in irreducible form. 
In Section \ref{sec:appl} we show some applications of our result, and  
in Section \ref{sec:concl} we conclude. The proof of a key technical lemma is postponed to \cref{app}.

\section{The irreducible form for MPS}
\label{sec:irr}

In this section we introduce the irreducible form and show its basic properties (\cref{ssec:irr}), and then study how it behaves under blocking (\cref{ssec:block}).

Let us first fix some notation. 
We are given a tensor $A = \left\{A^i\in \mathcal{M}_D\right\}_{i=1,\ldots, d}$, where $\mathcal{M}_D$ denotes the set of $D\times D$ complex matrices, and $D$ is called the bond dimension of $A$. 
This tensor defines the family of translationally invariant MPS of Eq.\ \eqref{eq:VA}. 
As in this paper we will only deal with translationally invariant MPS, we will often simply call them MPS. 
  Given  $A$, 
 the completely positive (CP) map $\E_{A}$ associated to it is defined as  $\E_{A}(\cdot) = \sum_{i=1}^d A^i\: \cdot\: A^{i \dagger}$, and its dual map 
as $\E^*_{A}(\cdot) = \sum_{i=1}^d  A^{i\dagger} \: \cdot\: A^i$.

\subsection{The irreducible form}
\label{ssec:irr}

The procedure that transforms $A$ to its irreducible form is  the same as the one that transforms it to its canonical form \cite{Ci15}, 
with the only difference that we do not block sites together. 
That is, we   project $A$ onto its invariant subspaces (as explained, e.g., at the beginning of  Sec.\ 2.3.\ of Ref.\ \cite{Ci15}), which allows to express $A$ in a block diagonal form 
$$
A^i=\bigoplus_{j\in\tilde{J}} {\tilde{A}}_j^i ,
$$ 
so that the CP map associated to each block $\E_{\tilde{A}_j}:\mathcal{M}_{D_j}\rightarrow \mathcal{M}_{D_j}$ is {\it irreducible} \cite{Wo11}, in the sense that
there exists no non-trivial hermitian projector $P$ so that  $\E_{\tilde{A}_j}(P\mathcal{M}_{D_j}P)\subseteq P\mathcal{M}_{D_j}P$.
Denoting  the spectral radius of $\E_{\tilde{A}_j}$ by $\varrho_j$ (recall that $\varrho_j>0$),  and defining $ {A}_j^i :=\tilde{A}^i_j/\sqrt{\varrho_j}$, 
we are left with 
\be 
\label{eq:irreducible-form}
A^i = \bigoplus_{j\in \tilde{J}} \mu_j  A_j^i ,
\ee
where $\mu_j := \sqrt{\varrho_j} >0 $, and 
every  $\E_{A_j}$ is an irreducible CP map of spectral radius 1. This is known to be equivalent to having 1 as a non-degenerate eigenvalue with an associated strictly positive eigenvector for both $\E_{A_j}$ and  $\E_{A_j}^*$ \cite{Wo11}. There could however be other eigenvalues of modulus one. If this is not the case the map is called \emph{primitive}, and the corresponding block or tensor is called {\it normal} \cite{Ci15}. We shall refer to the set of eigenvalues of magnitude one of a map as its peripheral spectrum, following \cite{Wo11}.  

If there are other eigenvalues of modulus one, one can show (see e.g. \cite{Wo11}) that they must be exactly the $m_j$-roots of unity for some $m_j$, each of them with multiplicity 1. 
This is why we shall call the blocks $A_j$ in the decomposition of Eq.~\eqref{eq:irreducible-form} {\it periodic}, and $m_j$ its \emph{period} or \emph{periodicity}. Note that normal tensors are just periodic tensors with period equal to $1$. 

\begin{definition}[Irreducible form]
A tensor $A$ is in irreducible form if it is in form \eqref{eq:irreducible-form} with $\mu_j>0$, and every $\E_{A_j}$ is an irreducible CP map of  spectral radius $1$.
\end{definition}

We have just shown the following. 

\begin{proposition}\label{thm:irr}
Given any tensor $A$, 
one can always find another tensor $B$ in irreducible form such that $\mc{V}(A) = \mc{V}(B)$. 
\end{proposition}

We have also seen that if $A$ has bond dimension $D$, then $B$ has bond dimension $\tilde{D}\leq D$. Moreover, if all blocks in~\eqref{eq:irreducible-form} are normal \cite{Ci15} (i.e.\ have period $1$), then \eqref{eq:irreducible-form} is just the canonical form of $A$ according to the definitions in \cite{Ci15}.

The next step is to group blocks that are {\it essentially the same} in the following sense.

\begin{definition}[Repeated  blocks] \label{def:repeated}
We say that two blocks, say $A_j$ and $A_k$, 
are \emph{repeated} if there exist a phase $\xi$ and an invertible matrix  $Y$ so that  
\be
A_j^i=e^{i\xi}YA_k^iY^{-1} .
\label{eq:rep} 
\ee
We say that they are \emph{equivalent} if \eqref{eq:rep} holds with $\xi=0$.
\end{definition}

Clearly, a block of periodicity $m$ cannot be repeated with one of
periodicity $n\ne m$. Taking any minimal set $J\subseteq \tilde{J}$ of
non-repeated blocks, there is a permutation of the blocks in
\eqref{eq:irreducible-form} that allows to rewrite $A$ as 
\be
A^i = 
\bigoplus_{j\in J} 
\left( R_j \otimes A^i_j\right) ,
\label{eq:bdnr}
\ee
where $R_j:= \diag(\mu_{j,1},\ldots, \mu_{j,r_j})$, with $\mu_{j,l}\in
\mathbb{C}\setminus\{0\}$. We shall henceforth refer to $\{\mu_{j,l}\}_{l}$ as the
multiplicities of block~$A_j$.

\begin{definition}[Basis of periodic tensors]
Any set $\{A_j\}_{j\in J}$ of non-repeated periodic tensors that allows to rewrite (up to a similarity transformation) a tensor $A$ in irreducible form as in \eqref{eq:bdnr} is called a \emph{basis of periodic tensors for $A$}. 
Two bases of periodic tensors are called \emph{equivalent} if they have the same number of elements, and each element of one basis is equivalent to an element of the other basis. Note that the prefactors of the basis elements can always be absorbed in the multiplicities. 
\end{definition}

If all blocks are normal in \eqref{eq:irreducible-form}, then the set $\{A_j\}_{j\in J}$ is called  instead a basis of normal tensors \cite{Ci15}.

\

As in the case of the canonical form, we can also impose conditions on the fixed points of the maps $\E_{A_j}$. That is,  let  $\rho_j$ denote the unique (positive definite) fixed point of  $\E^*_{A_j}$. 
Defining 
\be
A'^i_j := \sqrt{\rho_j} \: A^i_j \: \frac{1}{\sqrt{\rho_j}}\quad   \label{eq:unital}
\ee
we have that $\E^*_{A'_j}(\openone_{D_j})=\openone_{D_j}$, i.e.\ it is unital. 
(Note that this is equivalent to stating that $\E_{A_j}$ is trace-preserving.)
Since Eq.\ \eqref{eq:unital} is a similarity transform, $\E^*_{A'_j}$ is still irreducible and with spectral radius 1
(Proposition 6.6. of \cite{Wo11}), and so is $\E_{A'_j}$. 
Now let $\sigma_j$ denote the fixed point of $\E_{A'_j}$, which is again
positive definite, $\sigma_j>0$, and has spectral decomposition $\sigma_j = U_j \Lambda_j U_j^\dagger$. 
Defining 
$$
A''^i_j = U_j^\dagger \: A'^i_j \:U_j , 
$$
we have that $\E_{A''_j}(\Lambda_j) = \Lambda_j$ where $\Lambda_j$ is diagonal and positive definite.  
 Moreover $\E^*_{A''_j}$ is still unital. 
 
 That is, without loss of generality we can take every block in \eqref{eq:irreducible-form} (and \eqref{eq:bdnr}) with the additional property that the CP maps $\E_{A_j}$ are trace preserving and have a diagonal positive matrix as a unique fixed point. If this is the case, by analogy with \cite{Ci15} we say that $A$ is in \emph{irreducible form II}. 
 Note that if two blocks with associated trace-preserving CP maps are repeated, then the invertible matrix $Y$ in \cref{def:repeated} must be unitary. Note also that one can freely change from irreducible form to irreducible form II just by a block diagonal similarity transformation. Hence one can prove the results here in either form and they will immediately apply for the other one, just by replacing  appropriately ``invertible map'' by ``unitary map''.

Finally, $m_j$-periodic tensors $A_j$ are known to have \cite{Wo11} the following off-diagonal structure 
\be
A^i_j = \sum_{u=0}^{m_j-1} P_{j,u} \: A^i_j \: P_{j,u+1}, 
\qquad \textrm{with }\:\: P_{j,u} P_{j,u'}=\delta_{u,u'}P_{j,u}, \quad \sum_{u=0}^{m_j-1}P_{j,u}=\openone_{D_j}, \quad P_{j,m_j}= P_{j,0}.
\label{eq:Aoffdiag}
\ee
If the tensors are in irreducible form II, the projectors $P_{j,u}$ are hermitian, 
$\E^*_{A_j}(P_{j,u})=P_{j, u+1}$ and hence the unitary $U_j= \sum_{u=0}^{m_j-1} \omega_j^u P_{j,u}$ verifies that 
\be\label{eq:peripheral}
\E^*_{A_j}(U_j^l)=\omega_j^l U_j^l,\quad 
\textrm{with }\:\: \omega_j=e^{i\frac{2\pi }{m_j}}, 
\quad \textrm{for }\:\: l=0,\dots, m_j-1. 
\ee 
From \eqref{eq:Aoffdiag} it follows that 
$$
|V_N(A_j)\ra =
\begin{cases}
\sum_{u=0}^{m_j-1} \sum_{i_1,\ldots, i_N}  \tr(P_{j,u} \: A^{i_1}_j \cdots A^{i_N}_j)|i_1,\ldots, i_N\ra & \textrm{if } m_j \mid N\label{eq:VAj}\\
0 & \textrm{else} .
\end{cases}
$$
That is, if $m_j$ divides $N$, then $|V_N(A_j)\ra$ is a sum of $m_j$ $m_j$-periodic terms, and otherwise the state is 0. 
Thus, $|V_N(A_j)\ra$ is translationally invariant in a non-trivial way, since the translation operator by one site generates a cyclic permutation of the $m_j$ terms.


\subsection{Blocking periodic tensors}
\label{ssec:block}

Before proceeding, it is important to analyse how periodic tensors, and hence the irreducible form, behave under the blocking of tensors. 
Since blocking tensors corresponds to taking powers of the associated CP maps, the question is equivalent to studying how irreducible CP maps behave under  powers. For that, the off-diagonal decomposition \eqref{eq:Aoffdiag} plays a key role. 
 In this direction, the following result is proven in \cite{Ca13b}, 
 but we shall include here a proof  for completeness. 
 We will denote the tensor $A$ after blocking $p$ sites by $A^{(p)}$, i.e. 
$$
A^{(p)}=\left\{ A^{i_1}A^{i_2}\cdots A^{i_p}  \right\}_{i_1,\ldots, i_p\in \{1,\ldots d\}}\; .
$$
We will also denote by $\mathbf{i}$ the multiindex that contains $(i_1,\ldots, i_p)$.

 \begin{lemma}[Blocking a single periodic block]
\label{lem:bdcf}
Let $A$ be in irreducible form II with a single periodic block of
periodicity $m$.  If $p$ is a multiple of $m$, then $A^{(p)}$ is in
canonical form II with a basis of normal tensors given by $\{C_u=P_u
A^{(p)}\}_{u=0}^{m-1}$ (with $P_u$ as in \eqref{eq:Aoffdiag}, but without
the block index $j$).
\end{lemma}

Note that for convenience we choose not to project out the zero blocks in
$C_u$, this is, $C_u$ consists of the actual normal tensor, supported on
the range of $P_u$, padded with zeros. We will stick to this convention in
the following. In addition, we will also assume gauge transformations
acting between different $C_u$ and $C_v$ to be only defined on the
respective supports, and being zero outside.

\begin{proof}
Without loss of generality we can assume that $p=m$. We first want to see that  $\{C_u\}_{u=0}^{m-1}$ forms a basis of normal tensors. 
We will first show that each $C_u$ is a normal tensor, 
which amounts to seeing that $\E_{C_u}$ is primitive. 
Note first that since $A$ is a periodic block, $\E_A$ is irreducible with peripheral spectrum  $\{\omega^{r}\}_{r=0}^{m-1}$, where $\omega = e^{i2\pi/m}$.  
The CP map associated to $A^{(m)}$ is $\E_A^m$ (which denotes  the $m$-fold application of the map $\E_A$),
and it can be expressed as 
$$
\E_A^m(\rho) = \sum_{u,u'=0}^{m-1} P_u \E_A^m ( P_u\: \rho \: P_{u'}) P_{u'} .
$$
This map has $1$ (with multiplicity $m$) as a unique eigenvalue of magnitude one. 
Denote the fixed point of  $\E_{A}$ by $\Lambda_A$. 
Then it is immediate to see  that the set of fixed points of  $\E_A^m$ is given by $\{P_u\Lambda_A P_u\}_u$, the set of fixed points of $\E_A^{* m}$ is given by $\{P_u\}_u$, 
that  $P_u\Lambda_A P_u$ and $P_u$ are the fixed points of  $\E_{C_u}$ and $\E_{C_u}^*$,  respectively, and that $\E_{C_u}$ does not have any other eigenvalue of magnitude $1$. Therefore $\E_{C_u}$ is primitive, and  $C_u$ is a normal tensor in canonical form II, for all $u$. 

It only remains to be seen that the $C_u$'s are non-repeated. 
So imagine that there were a $C_u$ and a $C_v$ related to each other by
\be
C_u^{\mathbf{i}} = e^{i\xi} U C_v^{\mathbf{i}} U^\dagger . \label{eq:Cu}
\ee
As mentioned above, we choose $U$ such that it is only
non-zero on the respective supports, $U=P_uUP_v$, where it acts like a
unitary (i.e.\ $UU^\dagger=P_u$ and $U^\dagger U=P_v$). 
We have that 
$$
\E_A^m(U) = \sum_{\mathbf{i}} C_u^\mathbf{i}   U  C_v^{\mathbf{i} \dagger} = e^{i\xi}U  \sum_{\mathbf{i}} C_v^{\mathbf{i}}C_v^{\mathbf{i}\: \dagger}. 
$$
Noting that  $\E_{C_v}(\openone) =  \E_{C_v}(P_v) = P_v$ we obtain that $\E_A^m(U) = e^{i\xi}U$. \
But we established above that $\{P_u\Lambda_A P_u\}_{u=1}^m$ are the only   fixed points of $\E_A^m$, and that $\E_A^m$ has no other eigenvalues of modulus 1. 
Hence relation \eqref{eq:Cu} cannot hold.
\end{proof}

As a consequence, if a tensor $A$ is in irreducible form (irreducible form
II) and $m={\rm lcm}(\{m_j\}_{j\in J})$, the tensor $A^{(m)}$ is in
canonical form (canonical form II).

\cref{lem:bdcf} can be easily generalised to an arbitrary $p$. 

\begin{lemma}
\label{lem:blocking-arbitrary}
Let $A$ be in irreducible form II with a single periodic block of periodicity $m$. 
Take $p\in \mathbb{N}$ and let $r={\rm gcd}(m,p)$. Then $A^{(p)}$ is in irreducible form II, with $r$ periodic tensors of periodicity $m/r$, given by 
$C_\alpha= \tilde{P}_\alpha A^{(p)}$, $\alpha=0,\ldots, r-1$, where
$$
\tilde{P}_\alpha=\sum_{k=0}^{m/r-1} P_{[\alpha + pk]_m}, 
$$
where the notation $[s]_m$ stands for $(s\mod m)$.
\end{lemma}

Though we will not need it here, the tensors $C_\alpha$ are non-repeated. The argument is essentially the same as in the previous lemma. We just need the following trivial observation. 

\begin{lemma}\label{lem:unique-dec}
For each $u\in\{0,\ldots, m-1\}$ there exists a unique $\alpha_u\in\{0,\ldots, r-1\}$ and a unique $k_u\in\left\{0,\ldots, \left(\frac{m}{r}-1\right)\right\}$ so that
$u=[\alpha_u+ pk_u]_m$. 
\end{lemma}

This implies that in \cref{lem:blocking-arbitrary}, $\tilde{P}_\alpha \tilde{P}_{\beta}=\delta_{\alpha,\beta}\tilde{P}_\alpha$ and $\sum_\alpha \tilde{P}_\alpha=\1_{\mathcal{M}_D}$. Moreover, by their definition, $C_\alpha=\tilde{P}_\alpha A^{(p)}=\tilde{P}_\alpha A^{(p)}\tilde{P}_\alpha$ for all $\alpha$. Hence, each of the $C_\alpha$ defines a new block. We are thus only left to show that each of them is periodic of periodicity $m/r$. Since we know that the peripheral spectrum of $\E_{A}^p$ is $\{\omega^l\}_{l=0}^{\frac{m}{r}-1}$, with $\omega=e^{i\frac{2\pi}{m}}$, and each eigenvalue has multiplicity $r$, to finish the proof of \cref{lem:blocking-arbitrary}, it is enough to see that the peripheral spectrum of each $\E_{C_\alpha}$ is $\{\omega^l\}_{l=0}^{\frac{m}{r}-1}$. 
This, however, follws immediately using Eq.\ \eqref{eq:peripheral} since
$$
\E^*_{C_\alpha}(P_{[\alpha + pk]_m})=P_{[\alpha + p(k+1)]_m}\; .
$$

\section{The fundamental theorem for MPS}
\label{sec:thm}

In this section we present the fundamental theorem for MPS in irreducible form. 
In \cref{ssec:prelim} we will present some preliminary results, 
in \cref{ssec:thmprop} we will present the theorem in  the proportional case (i.e.\ when $|V_N(A)\ra$ is proportional to $|V_N(B)\ra$ for all $N$), and in 
\cref{ssec:thmequal} the theorem in the equal case  (i.e.\ when $\mc{V}(A)=\mc{V}(B)$). 

\subsection{Preliminary results}
\label{ssec:prelim}

We start by recalling the fundamental theorem for MPS in canonical form in the proportional case and in the equal case.

\begin{theorem}[Theorem 2.10.\ of \cite{Ci15}]
\label{thm:cf}
Let $A$ and $B$ be two tensors in canonical form with basis of normal tensors 
$\left\{A_j^i\right\}_{j\in J}$ and 
$\left\{B_k^i\right\}_{k\in K}$. 
If 
$|V_N(A)\ra$ and  $|V_N(B)\ra$ are proportional to each other for all $N$,  
then $|J|=|K|$, 
and for each $j$ there exists a $k$, a phase $\xi_k$, and an invertible  matrix $Y_k$ such that 
$A_j^i = e^{i\xi_k}Y_k B_k Y_k^{-1}.$  
That is, any two tensors in canonical form giving proportional MPS for all $N$ have equivalent bases of normal tensors.
\end{theorem}

\begin{theorem}[Corollary 2.11.\ of \cite{Ci15}] \label{cor:cfequal}
Let $A$ and $B$ be two tensors in canonical form so that $\mc{V}(A) = \mc{V}(B)$. 
Then $A$ and $B$ have equivalent bases of normal tensors with exactly the same multiplicities. In particular, there is an invertible $Y$ such that 
$A^i = YB^i Y^{-1}$.
\end{theorem}

This fundamental theorem is a direct consequence of the fact that normal
tensors are essentially {\it equal or orthogonal} in the following sense. 

\begin{proposition}[Lemma  A.2.\ of \cite{Ci15}] \label{equal-or-orthogonal}
Let $A_1$ and $A_2$ denote two normal tensors with bond dimensions $D_1, D_2$, and generating MPS $|V_N(A_1)\ra, |V_N(A_2)\ra$, respectively. 
Then
 \begin{subequations}
 \begin{align}
&& \lim_{N\to\infty} \la V_N(A_i) |V_N(A_i)\rangle &=1,\\
 \label{scprod}
&& \lim_{N\to\infty} |\langle V_N(A_1)|V_N(A_2)\ra |&=0 \text{ or } 1.
 \end{align}
 \end{subequations}
In the latter case, $D_1=D_2$ and there exists an invertible matrix $Y$ and a phase $\phi$ so that $A_1^i= e^{i\phi} Y  A_2^i  Y^{-1}$.
\end{proposition}

We will also require the following trivial result, as well as some results about moments of numbers that are presented below.

\begin{lemma}\label{Lem1t}
\label{Lem1}
Given $g\in \N$, there exists $\epsilon>0$ such that for any $m\ge g$, any $g$ vectors $|w_1\ra,\ldots |w_g\ra\in \mathbb{R}^m$ fufilling 
\begin{enumerate}
\item $|\la w_i|w_j\ra|\le \epsilon$ if $i\not =j$, and 
\item $|\la w_i|w_i\ra| \ge 1-\epsilon$ 
\end{enumerate}
must be linearly independent.
\end{lemma}

\begin{lemma}\label{lem:moments}
Consider two sets of non-zero complex numbers 
$\{\mu_l\}_{l=1}^s$, $\{\nu_l\}_{l=1}^t$. 
If
$$
\sum_{l=1}^s \mu_l^N = 
\sum_{l=1}^t \nu_l^N , \quad 1 \leq N\leq \mathrm{max}(s,t)\ , 
$$
then $s=t$, and there is a permutation $\pi$ such that 
$\mu_l =  \nu_{\pi(l)}$ for all $l$.
\end{lemma}

The proof of this lemma can be found e.g.\ in \cite{De15}.

\begin{lemma}\label{lem:momentsN0}
Consider two sets of non-zero complex numbers 
$\{\mu_l\}_{l=1}^s$, $\{\nu_l\}_{l=1}^t$. 
If there exists an $N_0$ such that
\be
\sum_{l=1}^s \mu_l^N = \sum_{l=1}^t \nu_l^N , \quad \forall N\geq N_0  \ ,
\label{eq:hypmoments}
\ee
then $s=t$, and 
there is a permutation $\pi$ such that
 $\mu_l = \nu_{\pi(l)}$ for all $l$.
\end{lemma}

\begin{proof}
We fix an arbitrary  $k\geq N_0$ and consider Eq.\ \eqref{eq:hypmoments}
for $N=nk$ with $n=1,\ldots, \mathrm{max}(s,t)$. 
By Lemma \ref{lem:moments}, $s=t$, and there is  a permutation $\pi_{k}$ such
that 
$\mu^{k}_l = \nu^{k}_{\pi_{k}(l)}$ for all $l$.
Since the number of possible permutations is finite, 
there must exist one permutation $\pi$ so that for two 
$k_1,k_2\geq N_0$ with $\textrm{gcd}(k_1,k_2) = 1$, 
\be
\mu_l^{k_1}  =
\nu_{\pi(l)}^{k_1} ,\quad 
\mu_{l}^{k_2} =\nu_{\pi(l)}^{k_2}\quad \forall l.
\label{eq:alphabeta}
\ee
(For example we could consider all prime $k$, and consider the permutations associated to them.)
By Bezout's identity, there exist integer numbers $a,b\in \mathbb{Z}$
so that $a k_1 +bk_2 =1.$
From this and Eq.\ \eqref{eq:alphabeta} it follows that $\mu_l = \nu_{\pi(l)}$ for all~$l$. 
\end{proof}

\subsection{The fundamental theorem for MPS -- proportional case}
\label{ssec:thmprop}

In this section we present the fundamental theorem for MPS in irreducible form in the proportional case. Throughout this section and the following one, we will denote the irreducible form of a tensor $B$ by 
$$
B^i= \bigoplus_{k\in \tilde{K}} \nu_k B^i_k.
$$
Additionally, $\{B_k\}_{k\in K}$ will denote a basis of periodic tensors for $B$, leading to 
\be
B^i= \bigoplus_{k\in {K}} \left( S_k\otimes B^i_k\right),
\label{eq:Bbdnr}
\ee
where $S_k := \diag(\nu_{k,1},\ldots, \nu_{k,s_k})$ with  $\nu_{k,l}\in
\mathbb{C}\setminus\{0\}$.  As in the case of \cref{thm:cf},  the fundamental theorem of
MPS in irreducible form will be a consequence of the following
generalisation of  \cref{equal-or-orthogonal}. 

\begin{proposition}\label{equal-or-orthogonal-generalized}
Consider two periodic tensors, $A$ and $B$, with corresponding bond dimension $D_a, D_b$, periods $m_a,m_b$ and generating MPS $|V_N(A)\ra, |V_N(B)\ra$. 
Then
$$
 \lim_{N \to \infty }
 \: \langle V_N(A)|V_N(A)\rangle=m_a
$$ 
(where the limit is restricted to multiples of $m_a$),
 and similarly for $B$.  Moreover, either
$$
 \lim_{N\to\infty} |\langle V_N(A)|V_N(B)\rangle|=0 ,
$$
(again with the limit restricted to common multiples of $m_a$ and $m_b$)
or $D_a=D_b$ and there exist a phase $\phi$ and an invertible matrix $Y$
so that $A^i= e^{i\phi} Y  B^i  Y^{-1}$; the latter can only happen if 
$m_a=m_b$.  As a consequence of this and  \cref{Lem1t},
given a basis of periodic tensors $\{A_j\}_{j\in J}$ for $A$,  there
exists an $N_0$ so that for all $N\ge N_0$, the set of non-zero vectors in
$\{|V_N(A_j)\ra\}_{j\in J}$ is linearly independent and spans linearly
$|V_N(A)\ra$.
\end{proposition}

Note that this last property is the one that justifies the name {\it basis of periodic vectors}. The proof of this result is postponed to \cref{app}.  We are now ready to present the main theorem in the proportional case. The proof mimics the one of Theorem \ref{thm:cf} given in \cite{Ci15}.

\begin{theorem}[Fundamental theorem for MPS -- proportional case]\label{thm:bd}
Let $A$ and $B$ be in irreducible form with basis of periodic tensors $\{A_j\}_{j\in J}$ and $\{B_k\}_{k\in K}$, respectively, and assume that for all $N$, 
$|V_N(A)\ra $ and $ |V_N(B)\ra$ are proportional to each other. 
Then for every $j\in J$ there is exactly one $k\in K$ (with the same period), a phase $\xi_k$, and 
 an invertible  matrix $Y_k$ such that 
$$
A_j^i = e^{i\xi_k} Y_k B_k^i Y_k^{-1}. 
$$
That is, any two tensors in irreducible form giving proportional MPS for
all $N$ have equivalent bases of periodic tensors.  
\end{theorem}

In the case of irreducible form II, both in  \cref{equal-or-orthogonal-generalized} and \cref{thm:bd} the invertible matrix must be unitary. 
Notice that this theorem only relates the bases of periodic tensors of $A$ and $B$, but it does not relate  $A$ and $B$ themselves, as the latter requires relating their multiplicities.  
This will only be possible in the equal case (\cref{thm:bdequal}).

\begin{proof}
Let us first consider $B_k$ for some given $k\in K$. It is not possible that $\langle V_N(B_k)|V_N(A_j)\rangle\to 0$ as $N\to\infty$ for all $j$, since otherwise the MPS generated by $A$ and $B$ could not be proportional for all $N$ (see \cref{Lem1}). Thus, according to  \cref{equal-or-orthogonal-generalized}, there must exist one $j_k\in J$ such that $B_k = e^{i\phi_k} Y_k A_{j_k} Y_k^{-1}$ for some phase $\phi_k$ and invertible matrix  $Y_k$. We thus conclude that $|K|\le |J|$.  But if we had considered $A_k$ to start with, we would obtain $|J|\le |K|$, so that $|J|=|K|$, and we are done.
\end{proof}

\subsection{The fundamental theorem for MPS -- equal case}
 \label{ssec:thmequal}

We are now ready to present the fundamental theorem for MPS in irreducible form in the equal case.

\begin{theorem}[Fundamental theorem for MPS -- equal case]\label{thm:bdequal}
Let $A$ and $B$ be in irreducible form.  If $\mc{V}(A) =  \mc{V}(B)$,
then $A$ and $B$ have equivalent bases of periodic tensors $\{A_j\}_{j\in
J}$ and $\{B_j\}_{j\in J}$. Moreover,  for each $j$, with periodicity
$m_j$, there exists a diagonal unitary matrix $Z_j$ with
$Z_j^{m_j}=\1_{r_j}$ 
(i.e.\ all diagonal entries are $m_j$-roots of unity), 
so that the diagonal matrices $R_j$ and $S_j$ associated to $A_j$ and $B_j$ in equations \eqref{eq:bdnr} and \eqref{eq:Bbdnr}, respectively, are related by $Z_jR_j=S_j$ after a reordering of the diagonal entries of $S_j$. 
Denoting $Z=\bigoplus_{j\in J} Z_j\otimes \1_{D_j}$, this implies that 
$$ 
Z A^i = YB^i Y^{-1}
$$ 
for all $i$, with $[A^i,Z]=0$ and $\mc{V}(A) = \mc{V}(ZA)$.
\end{theorem}

Note that any such $Z$ is a gauge degree of freedom in the MPS generated
by the tensor $A$, since the vector generated by block $A_j$,
$|V_N(A_j)\ra$, is non-zero only for $N$ multiple of $m_j$, and hence any
$m_j$-root of unity in the multiplicities of $A_j$ has no effect on the
state $|V_N(A)\ra$.

\begin{proof}
By Theorem \ref{thm:bd} we know that  both tensors have equivalent  bases
of periodic tensors, so that for all $j\in J$,  there is an invertible
$Y_j$ such that $B^i_j = Y_j A^i_{j} Y_j^{-1}$ (the phase can be absorbed
in $S_j$).  We thus have that $\vert V_N(A_j)\rangle = \vert
V_N(B_j)\rangle$. By assumption,
\[
\sum_j \mathrm{tr}(R_j^N) \vert V_N(A_j)\rangle = 
\sum_j \mathrm{tr}(S_j^N) \vert V_N(B_j)\rangle = 
\sum_j \mathrm{tr}(S_j^N) \vert V_N(A_j)\rangle\ ,
\]
and since by \cref{equal-or-orthogonal-generalized}, there is an $N_0$
such that the non-zero elements of $\{\vert V_N(A_j)\}_j$ are linearly
independent for $N\ge N_0$, we have that $\tr\left(R_j^N\right)=
\tr\left(S_j^N\right)$ for all $N\geq N_0$  such that $m_j \mid N$.
From Lemma \ref{lem:momentsN0}  we conclude that $r_j=s_j$ and that there
is a permutation $\pi$ such that $\mu_{j,l}^{m_j} = \nu_{j,\pi(l)}^{m_j}$
for all $l$, and thus 
$$
Z_j R_j =  T_j S_j T_j^\dagger,
$$
where $T_j$ is the matrix that implements the permutation $\pi$, 
and $Z_j$ has the desired properties. 
$Y$ is then constructed as $Y=\bigoplus T_j\otimes Y_j$ (and possibly permutations
of blocks), and $Z=\bigoplus Z_j\otimes \openone$.
\end{proof}


\section{Applications}
\label{sec:appl}

We now present some applications of the new fundamental theorem for MPS (\cref{thm:bdequal}). 
Namely, we first show an equivalence between the refinement properties of a state and the divisibility of its transfer matrix (\cref{ssec:div}), 
and then a more general characterisation of the tensors that give rise to MPS with symmetries (\cref{ssec:symm}).

\subsection{The refinement of a state and the divisibility of its transfer matrix}
\label{ssec:div}

In Ref.~\cite{De16}, certain continuum limits of translationally invariant
matrix product states are studied.  An indispensable tool for this study
is a relation between the refinement of a state  and the divisibility
properties of its transfer matrix (both concepts to be defined below),
which we establish in Theorem  \ref{thm:div} with the help of the new
fundamental theorem for MPS  (\cref{thm:bdequal}). 

We first need a couple of  definitions. 
A trace-preserving CP map 
$\E$ is called \emph{$p$-divisible} if there is another trace-preserving CP map $\E'$ such that $\E= \E'^p$ (where the latter denotes the $p$-fold application of the map). 
Given a tensor $B$, we say that  $\mc{V}(B)$ can be \emph{$p$-refined} if there exists another tensor $A$ and an isometry 
$W:\mathbb{C}^d\to (\mathbb{C}^d)^{\otimes p}$
such that 
\be
|V_{pN}(A)\ra = W^{\otimes N}|V_N(B)\ra \quad \forall N .
\label{eq:fine}
\ee

\begin{theorem}[$p$-refinement and $p$-divisibility] 
\label{thm:div}
Let $B$ be in irreducible form II. 
Then $\mc{V}(B)$ can be $p$-refined if and only if $\E_B$ is $p$-divisible. 
\end{theorem}

\begin{proof}
We first show that if $\mc V(B)$ can be $p$-refined, then $\mathcal
E_B$ is $p$-divisible. Our initial assumption is thus that for all $N$
\[
|V_{pN}(A)\ra = W^{\otimes N}|V_N(B)\ra =\vert V_N(C)\rangle,\ 
\]
where we have defined the tensor $C$ through
$$
C^{i_1,\ldots,i_p} := \sum_{i=1}^d W^{{(i_1,\ldots,i_p)},i}B^i.
$$
If we see $i_1,\ldots,i_p$ as a single physical index, it is clear that
$C$ is also in irreducible form II. 
Our goal is now to find an $\tilde A$ with $\vert V_{pN}(\tilde A)\ra =
\vert V_{pN}(A)\rangle$ such that $\mathcal E_{\tilde A}^p = \mathcal
E_C=\mathcal E_B$; in fact, we will construct an $\tilde A$ for which
even $\tilde A^{(p)}=C$.

First, by \cref{thm:irr}, we can assume without loss of generality that $A$ is in 
irreducible form II. Following \cref{lem:blocking-arbitrary}, $A^{(p)}$ is
then in irreducible form II as well.  Theorem \ref{thm:bdequal} then
implies that 
\[
ZA^{i_1}\cdots A^{i_p}=
ZA^{(p),i_1,\dots,i_p}=
UC^{i_1,\ldots,i_p}U^{\dagger}\ .
\]
In order to construct an $\tilde A$ such that  $\tilde A^{(p)}=C$, we thus
would like to distribute $Z$ evenly across all $A$'s and absorb it in $A$
($U$ is easily taken care of), which is subtle since $Z$ only commutes
with $A^{(p)}$.

So, to this end, consider a periodic block $A_j$  with periodicity $m_j$ in the
irreducible decomposition \eqref{eq:irreducible-form} of $A$, and let
$r_j={\rm gcd}(p,m_j)$. From \cref{lem:blocking-arbitrary}, we know
that the blocks that arise from $A_j$ in the irreducible decomposition of
$A^{(p)}$ are precisely $\tilde{P}_{j,\alpha} A_j^{(p)}$,
with
$\alpha=0,\ldots, r_j-1$, where
$$
\tilde{P}_{j,\alpha}=
\sum_{k=0}^{\frac{m_j}{r_j}-1} P_{j, [\alpha + pk]_{m_j}}\ ,
$$
with $P_{j,\alpha}$ defined in Eq.~\eqref{eq:Aoffdiag}.
Further, $Z$ (restricted to the support of $A_j$) acts on them as 
\begin{equation}
\label{eq:ZPA-is-cPA}
Z\tilde{P}_{j,\alpha} A_j^{(p)} =
c_{j,\alpha}\tilde{P}_{j,\alpha}
A_j^{(p)}\ ,
\end{equation}
where $c_{j,\alpha}^{m_j/r_j}=1$.
We now define
$$
A'^i_j=\sum_{u=0}^{m_j-1} d_{j,u} {P}_{j,u} A_j^i\ ,
$$
where $d_{j,u}={c_{j,\alpha_{u+1}}^{k_{u+1}}}/{c_{j,\alpha_u}^{k_u}}$,
with 
$\alpha_u\in\{0,\ldots,r_j-1\}$ and $k_u\in\{0,\ldots \tfrac{m_j}{r_j}-1\}$
such that  $u=[\alpha_u+pk_u]_m$ (those are unique,
cf.~\cref{lem:unique-dec}).
It follows that
\begin{equation}
\label{eq:Aprime-is-cPA}
A_j'^{(p)}=\sum_{u=0}^{m_j-1}\Bigg( \prod_{k=0}^{p-1}
d_{j,u+k}\Bigg) {P}_{j,u} A_j^{(p)}=
\sum_{u=0}^{m_j-1} c_{j,\alpha_u} {P}_{j,u} A_j^{(p)}\ ,
\end{equation}
where we have used
$\prod d_{j,u+k}=
{c_{j,\alpha_{u+p}}^{k_{u+p}}}/{c_{j,\alpha_u}^{k_u}}=c_{j,\alpha_u}$
(since by definition $k_{u+p}=k_u+1$ and $\alpha_{u+p}=\alpha_u$).
Comparing Eqs.~\eqref{eq:ZPA-is-cPA} and \eqref{eq:Aprime-is-cPA}, we
find that
for $A':=\bigoplus_{j\in\tilde{J}} \mu_j A'_j$ (cf.~\eqref{eq:irreducible-form}), 
we have that $A'^{(p)}=ZA^{(p)}$, and thus, defining $\tilde A=U^\dagger A' U$, we
have that $\tilde A^{(p)}=C$ and thus 
$\E_{\tilde A}^p=\E_C=\E_B$, completing the proof.

The proof of the converse is straightforward:  If $\E_B$ is $p$-divisible,
there exists an $A$ such that $\E_B=\E_A^p$.  Different Kraus
representations of the same channel are related by an isometry $W$ through
$$
A^{i_1}\cdots A^{i_p } = \sum_{i=1}^d W^{{(i_1,\ldots,i_p)},i}B^i\ , 
$$
which immediately implies \eqref{eq:fine}.
\end{proof}


The consequences of this theorem will be explored in \cite{De16}.


\subsection{Characterization of symmetries}
\label{ssec:symm}

For MPS in canonical form, the characterisation of the tensors that give rise to MPS with symmetries was studied in \cite{Pe08}. 
An extension of this characterisation to general MPS follows immediately from our new fundamental theorem (\cref{thm:bdequal}). 

\begin{corollary}[Symmetries]
Let $A$ be in irreducible form II.
If there is a local unitary matrix $u$ such that
$|V_N(A)\ra = u^{\otimes N} |V_N(A)\ra$ for all $N$, 
then there is a diagonal unitary $Z$ 
such that $[A^i,Z]=0$ for all $i$ and $\mc{V}(A) = \mc{V}(ZA) $, 
and another unitary $U$ such that for all $i'$
\be
\sum_{i}u^{i',i} A^i = Z U A^{i'} U^\dagger .
\label{eq:symm}
\ee
\end{corollary}

In \eqref{eq:symm} $u^{i',i}$ denotes the components of $u$. 
The consequences of this result will be explored elsewhere.

\section{Conclusions}
\label{sec:concl}

In this paper, we have introduced a new standard form for tensors
generating Matrix Product States, the \emph{irreducible form}. The
irreducible form generalizes the canonical form for MPS and is directly
applicable to general MPS, including those with non-trivial periodicity.
We have provided a constructive way to transform any MPS tensor $A$ into
a tensor $A'$ in irreducible form which generates the same family of MPS.
We have then derived a fundamental theorem for MPS in irreducible form ---
namely, we have shown that if two tensors $A$ and $B$ in irreducible form
give rise to the same family of translationally invariant MPS,
$\mc{V}(A)=\mc{V}(B)$, then these two tensors must be related by a
similarity transform $Y$ and a diagonal matrix of phases $Z$, namely $ZA^i
= YB^iY^{-1}$,  where $Z$ commutes with $A$, and is ``invisible'' to
the state, $\mc{V}(A)=\mc{V}(ZA)$.  This result generalizes the
fundamental theorem for MPS in canonical form, which yields the same
statement but without the $Z$. 

We have then presented two applications of this result. The first is a theorem
that proves that the refinement properties of a state are equivalent to
the divisibility properties of its associated quantum channel.  The second
is a characterisation of tensors that give rise to matrix product states
with symmetries.  Finally, our findings could also be applicable to further
scenarios where the fundamental theorem of MPS in canonical form has
proven useful, such as in the characterisation of 2D topological order
through Matrix Product Operators~\cite{Bu17}, or in the classification of
symmetry protected phases in one dimension~\cite{Ch11,Sc11}.

\subsection*{Acknowledgements} 
GDLC acknowledges support from the Elise Richter Fellowship of the FWF. 
This work was supported in part by the Perimeter Institute of Theoretical Physics through the Emmy Noether program. Research at Perimeter Institute is supported by the Government of Canada through Industry Canada and by the Province of Ontario through the Ministry of Economic Development and Innovation. 
DPG acknowledges support from  MINECO (grant MTM2014-54240-P), Comunidad de Madrid (grant QUITEMAD+-CM, ref. S2013/ICE-2801), and Severo Ochoa project SEV-2015-556. 
This project has received funding from the European Research Council (ERC) under the European Union's Horizon 2020 research and innovation programme (grant agreement GAPS No 648913).
NS acknowledges support by the European Union through the ERC-StG WASCOSYS (Grant No. 636201).
JIC acknowledges
support from the DFG through the NIM (Nano Initiative Munich).


\appendix


\section{Proof of  \cref{equal-or-orthogonal-generalized}}
\label{app}

Here we  prove \cref{equal-or-orthogonal-generalized} in the case of irreducible form II.  

First note that since $\la V_N(A)|V_N(A)\ra={\rm tr}(E_A^N)$, with
$E_A=\sum_{i=1}^d A^i\otimes \bar A^i$ the Choi matrix of $\mathcal E_A$
which has the same spectrum as $\mathcal E_A$, it is clear that $$
\lim_{N\to \infty}\la V_N(A) |V_N(A)\ra =m_a,
$$
for $N \in m_a \mathbb{N}$, 
and a similar equation holds for $B$.

For the rest, the first thing to prove is that if $A$ and $B$ have different periods $m_a\not = m_b$, then 
$$
 \lim_{N\to\infty} |\langle V_N(A)|V_N(B)\rangle|=0.
$$
Clearly, if $N$ is not a multiple of both $m_a$ and $m_b$, then $ |\la V_N(A)|V_N(B)\ra |=0$. So let us consider $N$ of the form $kp$ with $p={\rm lcm}(m_a,m_b)$. 
By \cref{lem:bdcf},  it is enough to show that $P_uA^{(p)}$ and $Q_v B^{(p)}$ cannot be related by 
\be
\label{eq:auxiliar-appendix}
 P_uA^{(p)} = e^{i\xi }U Q_{v} B^{(p)}U^{\dagger} .
 \ee
 For simplicity, we will present the argument for the case in which
$m_a=1, m_b>1$ (in which case $P_u=\openone$).  The general case is  analogous. 
 
Eq.\ \eqref{eq:auxiliar-appendix} implies that 
$$
\sum_{i_1\ldots i_{pN}} 
\tr(A^{i_1}\cdots A^{i_{pN}} ) |i_1\ldots i_{pN}\ra  =e^{i\xi N}
 \sum_{i_1\ldots i_{pN}} 
\tr(Q_{v} B^{i_1}\cdots B^{i_{pN}} ) |i_1\ldots i_{pN}\ra  
$$
for all $N$.
Applying the translation operator by one site to both sides of the equation, 
and using that the left hand side is translationally invariant and that $Q_{v}B^i = B^i Q_{v+1}$, 
we obtain that
$$
 \sum_{i_1\ldots i_{pN}} 
\tr(Q_{v} B^{i_1}\cdots B^{i_{pN}} ) |i_1\ldots i_{pN}\ra   = 
 \sum_{i_1\ldots i_{pN}} 
\tr(Q_{v+1} B^{i_1}\cdots B^{i_{pN}} ) |i_1\ldots i_{pN}\ra   
$$
for all $N$. 
But by \cref{lem:bdcf}, $Q_{v} B^{(p)}$ and 
$Q_{v+1} B^{(p)}$ are non-repeated blocks, 
and thus they cannot generate states equal to each other for all $N$, which is the desired contradiction.

So assume that the periods of $A$ and $B$ are both $m$. We consider two
separate cases. In the first case, there do not exist $u$, $v$ such that 
$$
 P_uA^{(m)} = e^{i\xi }U Q_{v} B^{(m)}U^{\dagger} .
$$
In this case, using \cref{equal-or-orthogonal} we obtain that 
$$
 \lim_{N\to\infty} |\la V_N(A)|V_N(B)\ra |=0.
$$ 
So we are only left with the case in which there exist $\tilde u,\tilde v\in
\{1,\ldots,m\}$, a phase $\lambda_{\tilde v}$ and a unitary $V$ such that \be
P_{\tilde u} A^{(m)}  = 
e^{i\lambda_{\tilde v}} V  Q_{\tilde v} B^{(m)}   V^{\dagger}.
\label{eq:ABm} 
\ee 
In the rest of the proof we will show that, in this case, 
there is a phase $\xi$ 
and a unitary matrix $U$ so that 
$A^i = e^{i\xi} U B^i U^{\dagger}. $ 
This will finish the proof of \cref{equal-or-orthogonal-generalized}.

First note that since $P_{\tilde u}$ and $Q_{\tilde v}$ are projectors,
Eq.\ \eqref{eq:ABm} implies that we can choose $V = P_{\tilde u}VQ_{\tilde
v}$.  (As earlier, $V$ is a unitary if restricted to its support and range.)
Eq.~\eqref{eq:ABm} also implies that for all $N$
\be
\sum_{i_1,\ldots, i_{mN}} \tr(P_{\tilde u} A^{i_1} \cdots A^{i_{mN}} ) |i_1,\ldots, i_{mN}\ra 
&=&
e^{i\lambda_{\tilde v} N}\sum_{i_1,\ldots, i_{mN}} 
		\tr(Q_{\tilde v}B^{i_1} \cdots B^{i_{mN}} ) |i_1,\ldots, i_{mN}\ra .\quad 
\label{eq:Nm}
\ee
If $T$ is the translation operator, applying $T^l$ for $l=1,2,\ldots, m-1$ on both sides of Eq.\ \eqref{eq:Nm}, we obtain that 
$$
\sum_{i_1,\ldots, i_{mN}} \tr(P_{\tilde u+l} A^{i_1} \cdots A^{i_{mN}} ) 
    |i_1,\ldots, i_{mN}\ra 
= 
e^{i\lambda_{\tilde v} N} \sum_{i_1,\ldots, i_{mN}} 
	\tr(Q_{\tilde v+l}B^{i_1} \cdots B^{i_{mN}} ) |i_1,\ldots, i_{mN}\ra . \qquad 
$$
We have by \cref{lem:bdcf} that $P_{\tilde u+l} A^{(m)}$ and $Q_{\tilde
v+l} B^{(m)}$ are normal tensors, and thus from \cref{thm:cf} it follows
that for every $l=1,\ldots, m-1$ there is a phase $\lambda_{\tilde v+l}$
and a unitary $U_{\tilde v+l}=P_{\tilde u+l} U_{\tilde v+l} Q_{\tilde
v+l}$ such that 
\be
&&P_{\tilde u+l} \: A^{i_1} \cdots A^{i_m}  = e^{i\lambda_{\tilde v+l}}
U_{\tilde v+l}\: Q_{\tilde v+l} \: B^{i_1} \cdots B^{i_m}U_{\tilde v+l}^{\dagger} .
\label{eq:blockedABprop}
\ee
Thus, for any $u=\tilde u+l$ and $v=\tilde v+l$, we have that $v-u=\tilde
u-\tilde v=:q$, i.e., $u$ and $v$ are related through a cyclic permutation.
Thus, throughout the rest of the proof we will use that 
\be
v= u+q \mod m. \label{eq:vprop}
\ee

Now, define the tensors
\be
&&A^i_u := P_u A^{i} P_{u+1}, \label{eq:Auprop}\\ 
&&B_{v}^i := U_v  Q_v B^{i} Q_{v+1} U_{v+1}^{\dagger},  \label{eq:Cvprop}
\ee
so that Eq.\ \eqref{eq:blockedABprop} reads
\be \label{eq:BCmprop}
&&A_u^{i_1}A_{u+1}^{i_2}\cdots A_{u+m-1}^{i_m} = 
e^{i\lambda_{v}}
B_{v}^{i_1} B_{v+1}^{i_2}\cdots B_{v+m-1}^{i_m}, \quad  \textrm{for } u=1,\ldots, m. 
\ee
By Lemma \ref{lem:bdcf}, $A_u^{i_1}\cdots A_{u+m-1}^{i_m}$  is normal for every $u$, and 
thus there is some length $N_0$ at which it becomes injective~\cite{Ci15}. 
Consider the tensor 
\be
F_u^\mathbf{i} := 
A_u^{i_1}A_{u+1}^{i_2}\cdots A_{u+m-1}^{i_{mN_0}} ,
\label{eq:Fu}
\ee
where $\mathbf{i}=(i_1,\dots,i_{mN_0})$. Since on the right hand side of
Eq.~\eqref{eq:Fu}, $A_u^{i_1}\cdots A_{u+m-1}^{i_m}$ is repeated $N_0$
times, $F_u^\mathbf{i}$ is injective. Thus there is an inverse 
 $\Omega_u$ such that 
\begin{equation}
\sum_{\mathbf{i}} (\Omega_u^\mathbf{i})_{\alpha,\beta} 
(F_u^\mathbf{i})_{\alpha',\beta'} =
\delta_{\alpha,\alpha'}\delta_{\beta,\beta'}, \quad u=1,\ldots, m.
\label{eq:Omegauprop}
 \end{equation}
Now consider the concatenation of tensors
$$
A_u^{i_1} F_{u+1}^{\mathbf{i}_1} A_{u+1}^{i_2} F_{u+2}^{\mathbf{i}_2} \cdots A_{u+m-1}^{i_{m-1}} F_{u}^{\mathbf{i}_{m-1}}.
$$
Note that this is simply the tensor $A^{i_1}_u\cdots A^{i_m}_{u+m-1}$
repeated $mN_0+1$ times, which by Eq.\ \eqref{eq:BCmprop} equals
$e^{i\lambda_v(mN_0+1)}$ times $(B^{i_1}_u\cdots B^{i_m}_{u+m-1})$
repeated $mN_0+1$ times.  We apply the inverses 
$$
\Omega_{u+1}^{ \mathbf{i}_1}\otimes \Omega_{u+2}^{ \mathbf{i}_2} \otimes
\cdots \otimes \Omega_{u+m-1}^{ \mathbf{i}_{m-1}}
$$
to it. Using \eqref{eq:Omegauprop} we obtain 
$A_u^{i_1}\otimes A_{u+1}^{i_2}\otimes \ldots  A_{u+m-1}^{i_m},$ and using 
Eq.\ \eqref{eq:BCmprop}, we obtain
$e^{i\eta_v}
B_v^{i_1}\otimes B_{v+1}^{i_2}\otimes \ldots  B_{v+m-1}^{i_m} $, 
where
$$
\eta_v := \lambda_v (mN_0+1) - N_0\sum_{l=1}^m \lambda_l. 
$$
That is, we have found that 
\be
A_u^{i_1}\otimes A_{u+1}^{i_2} \otimes \ldots  A_{u+m-1}^{i_m}
= 
e^{i\eta_v} B_{v}^{i_1}\otimes  B_{v+1}^{i_2}  \otimes \ldots  B_{v+m-1}^{i_m}  \qquad \textrm{for } u=1,\ldots, m .
\label{eq:resultprop}
\ee  

Applying  the translation operator $l$ times to  both sides of Eq.\ \eqref{eq:resultprop}, for $l=1,\ldots, m-1$, and using that Eq.\ \eqref{eq:resultprop} is valid for all $i_1,\ldots, i_m$, we obtain that $\eta_v$ is independent of $v$, 
which we will henceforth call  $\eta:=\eta_v$. 
Now, \eqref{eq:resultprop} gives
\be
A_{u}^{i} = \kappa_v  e^{i\eta/m}  \: B_{v}^{i} , \quad \textrm{for }u=1,\ldots, m, 
 \label{eq:thetaACprop}
\ee
with
\be
 \: \prod_{v=1}^{m} \kappa_v = 1.
\label{eq:prodkappaprop}
\ee
Using now that the operator norm $\left\|\sum_i A_u^{i \dagger} A_u^i\right\|=
\left\|\sum_i B_v^{i \dagger} B_v^i\right\|=1$, we obtain that  $|\kappa_v|=1$ for all $v$.  Thus $\kappa_v$ can be written as $\kappa_v = e^{i\theta_v} $
with $\theta_v\in \mathbb{R}$.

Therefore, recalling the definitions of $A_u$ (Eq.\ \eqref{eq:Auprop}) and $B_v$ (Eq.\ \eqref{eq:Cvprop}),  Eq.\ \eqref{eq:thetaACprop} can be written as 
\be
P_uA^i P_{u+1}   = e^{i\theta_v} 
e^{i\eta/m}
U_v Q_{v} B^i Q_{v+1} U_{v+1}^{\dagger}, \quad  \textrm{for } u=1,\ldots, m. 
\ee
Additionally, from Eq.\ \eqref{eq:prodkappaprop}  we have that $\sum_{v=1}^m \theta_v =0 \mod 2\pi$. 
Thus there are some other phases $\{\phi_v\}_{v=1}^m$ such that 
$\theta_v = \phi_v -\phi_{v+1}\mod 2\pi$ for $v=1,\ldots, m, $
where the sum $v+1$ is modulo $m$.
These can be chosen as 
\be
&& \phi_1=0, \quad
 \phi_v = - \sum_{l=1}^{v-1} \theta_l,
  \qquad \textrm{for }v=2,\ldots, m.
\ee
With this definition, we have that 
\be
 P_uA^i P_{u+1}   
&=& e^{i\eta/m}e^{i\phi_v}U_v Q_{v} B^i Q_{v+1} U_{v+1}^{\dagger} e^{-i\phi_{v+1}} 
, \qquad \textrm{for }u=1,\ldots, m.
\label{eq:result}
\ee
Thus, finally, defining 
$U  =  \sum_{u=1}^{m} e^{i\phi_{u+q}} P_{u}  U_{u+q} Q_{u+q}$ 
(where we have used the definition of $v$, Eq.\ \eqref{eq:vprop}), 
and the form of $A$ and $B$), Eq.\ \eqref{eq:result} can be written as 
\be
 A^i = e^{i\xi} U B^i U^\dagger,  
\ee
where $\xi =\eta/m $.


\end{document}